\documentclass[12pt]{article}
\usepackage{eurosym}
\usepackage{amssymb}
\usepackage{amsmath}
\usepackage{cite}

\newtheorem{theorem}{Theorem}

\newtheorem{corollary}[theorem]{Corollary}

\newenvironment{proof}[1][Proof]{\noindent\textbf{#1.} }{\ \rule{0.5em}{0.5em}}
\begin{document}

\title{Geometric Linearization for Constraint Hamiltonian Systems}
\author{Andronikos Paliathanasis $^{1,2,3}$ \thanks{%
Email: anpaliat@phys.uoa.gr} \\
{\ } {\ \textit{$^1$ Institute of Systems Science, Durban University of
Technology }}\\
{\ \textit{PO Box 1334, Durban 4000, Republic of South Africa}}\newline
\\
{\ } {\ \textit{$^2$ School for Data Science and Computational Thinking, Stellenbosch University}}\\
{\ \textit{44 Banghoek Rd, Stellenbosch 7600, South Africa}}\newline
\\
{\ }\textit{$^3$ Departamento de Matem\'{a}ticas, Universidad Cat\'{o}lica del
Norte,} \\
\textit{Avda. Angamos 0610, Casilla 1280 Antofagasta, Chile}}
\maketitle

\begin{abstract}
This study investigates the geometric linearization of constraint
Hamiltonian systems using the Jacobi metric and the Eisenhart lift. We
establish a connection between linearization and maximally symmetric
spacetimes, focusing on the Noether symmetries admitted by the constraint
Hamiltonian systems. Specifically, for systems derived from the singular
Lagrangian
\begin{equation*}
L\left( N,q^{k},\dot{q}^{k}\right) =\frac{1}{2N}g_{ij}\dot{q}^{i}\dot{q}%
^{j}-NV(q^{k}),
\end{equation*}%
where $N$ and $q^{i}$ are dependent variables and $\dim g_{ij}=n$, the
existence of $\frac{n\left( n+1\right) }{2}$ Noether symmetries is shown to
be equivalent to the linearization of the equations of motion. The
application of these results is demonstrated through various examples of
special interest. This approach opens new directions in the study of
differential equation linearization.

Keywords: Constraint Hamiltonian systems; Lie symmetries; Linearization;
\end{abstract}

\section{Introduction}

\label{sec1}

The theory of invariant infinitesimal transformations, known as Lie symmetry
analysis \cite{lie1,lie2,lie3}, is a powerful tool for the analytic
treatment of nonlinear differential equations. The concept of symmetry is
based on the existence of invariant functions in the solution space of the
differential equations \cite{ibra}. These invariant functions are mainly
used to simplify differential equations and to derive solutions \cite%
{Stephani,olver,Bluman}. Lie symmetry analysis systematically treats
dynamical systems and has been widely explored in various areas of applied
mathematics, see for instance \cite%
{sw7,sw3,hl0,hl0a,sw15,ami2,hl1,b2,b3,oli,prd,sja} and references therein..

Symmetry in physics is crucial because it is related to conservation laws
\cite{sa1,sa2,sa3}. The conservation laws of momentum for free particles,
energy for conservative systems, Kepler's laws of motion, and many others
are connected to infinitesimal transformations that keep the differential
equations invariant \cite{Stephani}. The existence of a sufficiently large
number of Lie symmetries for a given differential equation allows for
solving the equation by repeated order reduction using similarity
transformations or by determining a sufficient number of first integrals
\cite{olver}. Additionally, Lie symmetries can be used to classify dynamical
systems through their admitted invariant functions \cite{Bluman}.

This latter characteristic is important for the concept of global
linearization of differential equations via point transformations.
Specifically, nonlinear differential equations can be linearized if they
admit a specific number of Lie symmetries \cite{ibra}, as will be detailed
in the following sections. In this study, we focus on the global
linearization of a family of constraint Hamiltonian systems. We derive a set
of geometric criteria that allow us, for the first time, to express
nonlinear dynamical systems with fewer symmetries in a linear form.

We employ two different methods for the geometrization of non-relativistic
dynamical systems with conservative forces, the Einsenhart-Duval lift and
the Jacobi metric. In the Jacobi metric approach the conservative force is
absorbed in the kinetic terms, such that the dynamical effects of the force
to be part of the new geometry which describes the dynamical system. This is
achieved by a reparametrization of the independent variables. There are
various applications of Jacobi metric for the study of dynamical systems
\cite{jm0,jm1,jm2}, while the method has been generalized for the
description of relativistic systems \cite{jm3,jm4}. For a more general
discussion we refer the reader to \cite{pr1}.

On the other hand, in the Eisenhart-Duval method the conservative forces are
embedded into an extended geometry by introducing new independent variables
and degrees of freedom. The method was introduced by Eisenhart in early of
the twentieth century \cite{el1} and it was rediscovered latter by Duval et
al. in \cite{duval}. The Eisenhart-Duval method is the nonrelativistic case
of the Kaluza-Klein framework \cite{kk1,kk2}. Due to the simplicity of the
method, it has been widely applied for the study of nonlinear dynamical
systems. The Eisenhart lift of two-dimensional mechanical systems with or
without varying mass term is discussed in \cite{lf1}. The superintegrability
property of three-dimensional Newtonian systems investigated by using the
Eisenhart lift in \cite{lf1a}. In \cite{lf1b} the Eisenhart lift applied for
the study of the Toda chain for the study of nearest neighbor interacting
particles on a line. The geodesic description of the two fixed centers
problem by using the Eisenhart lift was studied in \cite{lf1c}.

Nevertheless, the Eisenhart lift has been applied not only in classical
mechanics but also for the analysis of the Schr\"{o}dinger \cite{lf3,lf4}
and of the Dirac \cite{lf5} equations of quantum mechanics. There is a
plethora of applications of the Eisenhart lift in relativistic physics and
cosmology, see for instance \cite{lf6,lf7,lf7,lf8} and references therein.
The extended minisuperspace by the Eisenhart lift for the study of quantum
cosmology introduced in \cite{lf9}. Recently, in \cite{lf10} the Eisenhart
lift employed for the derivation of the analytic solutions for the field
equations in scalar field cosmology. It was found that the field equations
can be linearized in the framework of minimally coupled scalar field
cosmological theory for a Friedmann--Lema\^{\i}tre--Robertson--Walker
background geometry and for an exponential potential for the description of
the mass for the scalar field. Moreover, in \cite{lf11} it was found that
the cosmological constant in Friedmann--Lema\^{\i}tre--Robertson--Walker
geometry can be recovered by mean of the application of the hidden
symmetries of for the extended minisuperspace in quantum cosmology.

Furthermore, \cite{lf2} the Eisenhart lift employed to write the equation of
motion for the oscillator in terms of the free particle. That is a well
known result provided by the symmetry analysis which relates the two
dynamical systems by sharing the common symmetry group \cite{fam}. We remark
that by geometrize a given dynamical systems it is feasible to employ
important results from the differential geometry in a systematic way. This
property is considered in the following in order to perform a geometric
linearization for a class of constraint Hamiltonian systems. We shall see
that there exist a relation between the geometric properties of the extended
geometry in the Eisenhart lift with the admitted symmetries of the original
system. The structure of the paper is as follows:

In Section \ref{sec2}, we discuss the mathematical tools primarily applied
in this study. In particular we discuss the Lie symmetry analysis for
differential equations and we focus in the case of second-order
differential. We also present Noether's theorems, which play an important
role in the main analysis of this study. Furthermore, in Section \ref{sec2a}%
, we provide a literature review on the global linearization of differential
equations through Lie symmetry analysis.

The family of dynamical systems considered in this work are introduced in
Section \ref{sec3}. Specifically, we examine a family of dynamical systems
described by a singular Lagrangian function, leading to equations of motion
where the Hamiltonian function is constrained. Additionally, we discuss in
detail two geometrization approaches for dynamical systems within this
family. Section \ref{sec4} includes the main results of this study, where we
derive a new set of geometric conditions and criteria under which constraint
Hamiltonian systems can be transformed into an equivalent system of a free
particle in flat space, thereby becoming linearizable.

We demonstrate the application of this geometric approach in Section \ref%
{sec5}, where we explore the linearization of nonlinear dynamical systems of
special interests in gravitational physics. The geometric linearization of
the Szekeres system with or without the cosmological constant is discussed.
Furthermore we investigate the linearization of the minisuperspace for the
gravitational model in a static spherically symmetric spacetime with a
charge, that is, the exact solution for the Reissner-Nordstr\"{o}m black
hole metric can be constructed by the solution of the free particle. Last
but not least, dynamical systems with interest in Newtonian mechanics are
discussed. Finally, in Section \ref{sec6}, we summarize our results.

\section{Preliminaries}

\label{sec2}

In this Section, we briefly discuss the basic mathematical elements
necessary for this study.

\subsection{Lie symmetries of differential equations}

In the following lines, we review the basic definitions for the Lie symmetry
analysis of ordinary differential equations.

Let us assume the $n$-dimensional dynamical system of $\mu$th-order
differential equations of the form
\begin{equation}
q^{\left( \mu \right) i}=\omega ^{i}\left( t,q^{k},\dot{q}^{k},\ddot{q}%
^{k},...,q^{\left( \mu -1\right) i}\right)  \label{Ls.01}
\end{equation}%
in which $t$ is the independent variables, and $q^{i}$ are the dependent
variables, $q^{i}=\left( q^{1},q^{2},...,q^{N}\right) $. Moreover, a dot
represents the total derivative with respect to the independent variable $t$%
, i.e.,
\begin{equation}
\dot{q}^{i}=\frac{dq^{i}}{dt},~\ddot{q}^{i}=\frac{d^{2}q^{i}}{dt^{2}}%
,~...,~q^{\left( \mu \right) }=\frac{d^{\mu }q}{dt^{\mu }}.
\end{equation}

Consider the infinitesimal transformation
\begin{eqnarray}
\bar{t} &=&t+\varepsilon \xi \left( t,q^{k}\right)  \label{Ls.02} \\
\bar{q}^{i} &=&q^{i}+\varepsilon \eta ^{i}\left( t,q^{k}\right)
\label{Ls.03}
\end{eqnarray}%
with generator the vector field $X=\xi \left( t,q^{k}\right) \partial
_{t}+\eta ^{i}\left( t,q^{k}\right) \partial _{i}$.

The vector field $q^{[\mu ]}$ is the $\mu $th-prolongation of $q$ in the Jet
space $J_{V}=\left\{ t,q^{i},\dot{q}^{i},\ddot{q}^{i},...,q^{\left( \mu
\right) }\right\} $ defined as%
\begin{equation}
X^{\left[ \mu \right] }=X+\eta _{\left[ 1\right] }^{i}\partial _{\dot{q}%
^{i}}+\eta _{\left[ 2\right] }^{i}\partial _{\ddot{q}^{i}}+...+\eta _{\left[
\mu \right] }^{i}\partial _{q^{\left( \mu \right) i}},
\end{equation}%
in which~$\eta _{\left[ 1\right] }^{i}$,~ $\eta _{\left[ 2\right] }^{i}$,
...,$~\eta _{\left[ \mu \right] }^{i}$ are given by the following expressions%
\begin{eqnarray}
\eta _{\left[ 1\right] }^{i} &=&\dot{\eta}^{i}-\dot{q}^{i}\dot{\xi},
\label{ls.04} \\
\eta _{\left[ 2\right] }^{i} &=&\dot{\eta}_{\left[ 2\right] }^{i}-\ddot{q}%
^{i}\dot{\xi},  \label{ls.05} \\
&&...  \label{ls.06} \\
\eta _{\left[ \mu \right] }^{i} &=&\dot{\eta}_{\left[ \mu -1\right]
}^{i}-q^{\left( \mu \right) i}\dot{\xi}.  \label{ls.07}
\end{eqnarray}

Then we say that the system of differential equations (\ref{Ls.01}) will be
invariant under the application of the infinitesimal transformations (\ref%
{Ls.02}), (\ref{Ls.03}) if and only if there exists a function $\lambda $
such that the following condition holds \cite{ibra}

\begin{equation}
\left[ X^{\left[ \mu \right] },A\right] =\lambda A,  \label{ls.08}
\end{equation}%
where%
\begin{equation}
\left[ X^{\left[ \mu \right] },A\right] =X^{\left[ \mu \right] }A-AX^{\left[
\mu \right] }~,
\end{equation}
and operator $A$ is Hamilton's vector%
\begin{equation}
A=\frac{\partial }{\partial t}+\dot{q}^{i}\frac{\partial }{\partial q^{i}}+%
\ddot{q}^{i}\frac{\partial }{\partial \dot{q}^{i}}+...+\omega ^{i}\left(
t,q^{k},\dot{q}^{k},\ddot{q}^{k},...,q^{\left( \mu -1\right) i}\right) \frac{%
\partial }{\partial q^{\left( \mu \right) i}}.  \label{ls.09}
\end{equation}

Condition (\ref{ls.08}) is known as the Lie symmetry condition, and $X$ is a
Lie point symmetry for the dynamical system (\ref{Ls.01}).

Assume that $f^{i}$ is a solution vector for the dynamical system (\ref%
{Ls.01}), that is $Af^{i}=0$, then, the Lie symmetry condition (\ref{ls.08})
becomes $\left[ X^{\left[ \mu \right] },A\right] f=\lambda Af,$ that is,
\cite{ibra}%
\begin{equation}
X\left( Af^{i}\right) =0,  \label{ls.10}
\end{equation}%
or equivalently%
\begin{equation}
\eta _{\left[ \mu \right] }^{i}=X^{\left[ \mu -1\right] }\omega ^{i}\left(
t,q^{k},\dot{q}^{k},\ddot{q}^{k},...,q^{\left( \mu -1\right) i}\right) .
\label{ls.11}
\end{equation}%
The solution of the latter linear system defines the functional form of the
generator vector $X$ for the infinitesimal transformation (\ref{Ls.02}), (%
\ref{Ls.03}).

\subsection{Second-order differential equations}

In this study, we focus on second-order differential equations of the form
\cite{Stephani}
\begin{equation}
\ddot{q}^{i}=\omega ^{i}\left( t,q^{k},\dot{q}^{k}\right) \text{.}
\label{ls.12}
\end{equation}

Therefore, the components $\eta ^{\left[ 1\right] }$,~$\eta ^{\left[ 2\right]
}$ of the second-prolongation reads \cite{Stephani}%
\begin{equation}
\eta ^{\left[ 1\right] }=\eta _{,t}^{i}+q^{i}\left( q^{k}\eta _{,k}^{i}-\xi
_{,t}\right) -\dot{q}^{i}\dot{q}^{k}\xi _{,k},  \label{ls.13}
\end{equation}%
\begin{equation}
\eta _{\left[ 2\right] }^{i}=\eta _{,tt}^{i}+2\left( \eta _{,tk}^{i}-\xi
_{,tt}\right) \dot{q}^{k}+\left( \eta _{,kr}^{i}-2\xi _{,tq}\right) \dot{q}%
^{k}\dot{q}^{r}-\dot{q}^{i}\dot{q}^{k}\dot{q}^{r}\xi _{,kr}+\left( \eta
_{,k}^{i}-2\xi _{,t}-3\xi _{,kr}\dot{q}^{r}\right) \ddot{q}^{k}.
\label{ls.15}
\end{equation}%
thus the symmetry conditions (\ref{ls.11}) becomes \cite{Stephani}%
\begin{eqnarray}
0 &=&\eta _{,tt}^{i}+2\left( \eta _{,tk}^{i}-\xi _{,tt}\right) \dot{q}%
^{k}+\left( \eta _{,kr}^{i}-2\xi _{,tq}\right) \dot{q}^{k}\dot{q}^{r}  \notag
\\
&&-\dot{q}^{i}\dot{q}^{k}\dot{q}^{r}\xi _{,kr}+\left( \eta _{,k}^{i}-2\xi
_{,t}-3\xi _{,kr}\dot{q}^{r}\right) \omega ^{i}\left( t,q^{k},\dot{q}%
^{k}\right)  \notag \\
&&-\omega _{,t}^{i}-\omega _{,k}^{i}\eta ^{k}-\omega _{,\dot{q}%
^{j}}^{i}\left( \eta _{,t}^{j}+q^{j}\left( q^{k}\eta _{,k}^{j}-\xi
_{,t}\right) -\dot{q}^{j}\dot{q}^{k}\xi _{,k}\right) .  \label{ls.16}
\end{eqnarray}

Lie symmetries have numerous applications. They are used to derive invariant
functions and conservation laws. Furthermore, they are applied to categorize
differential equations and establish criteria for when differential
equations are equivalent to a linear differential equation under a point
transformation. In the following lines, we discuss Noether's theorems for
the construction of conservation laws and geometric linearization through
symmetries.

\subsection{Noether's theorems}

Let's turn our attention to the case where the second-order dynamical system
(\ref{ls.12}) arises from the variation of the action integral
\begin{equation}
S=\int L\left( t,q,\dot{q}\right) dt,  \label{ls.17}
\end{equation}%
where $L\left( t,q,\dot{q}\right) \,$\ is the so-called Lagrangian function.

In 1918 \cite{noe}, Emmy Noether published two groundbreaking theorems that
relate the symmetries of the variation principle to conservation laws in
dynamical systems.

The first theorem states that if, under the application of the infinitesimal
transformations (\ref{Ls.02}), (\ref{Ls.03}), there exists a function $f$
such that the following condition holds
\begin{equation}
X^{\left[ 1\right] }L+L\frac{d\xi }{dt}=\frac{df}{dt},  \label{Ns.03}
\end{equation}%
then $X$ is called a Noether symmetry. It is evident that Noether symmetries
are Lie symmetries for the dynamical system (\ref{ls.12}), but the converse
is not necessarily true; that is, Noether symmetries form a subalgebra of
the Lie symmetries for the dynamical system. The function $f$ in (\ref{Ns.03}%
) represents a boundary term introduced to account for infinitesimal changes
in the action integral due to infinitesimal changes in the boundary of the
domain, caused by infinitesimal transformations of the variables in the
action integral.

Noether's second theorem provides a systematic way to derive conservation
laws. In particular, if $X=\xi \partial _{t}+\eta \partial _{i}$ is a
Noether symmetry for the dynamical system (\ref{ls.12}), then the function
\begin{equation}
\Phi \left( t,q^{k},\dot{q}^{k}\right) =\xi \left( \dot{q}^{i}\frac{\partial
L}{\partial \dot{q}^{i}}-L\right) -\eta ^{i}\frac{\partial L}{\partial q^{i}}%
+f,  \label{Ns.04}
\end{equation}%
is a conservation law, that is, $A\left( \Phi \right) =0$; consequently $%
X\left( \Phi \right) =0$.

For a recent extended discussion of Noether's theorems we refer the reader
to \cite{noe1}.

\section{Linearization through symmetries}

\label{sec2a}

One of the significant findings of Lie symmetry analysis is the
linearization criterion for a second-order ordinary differential equation.
The linearization process is crucial because it provides a method to express
the analytic solution in terms of simple functions, thereby offering a
better understanding of the dynamics of the dynamical system.

Sophus Lie's theorem states that if a second-order ordinary differential
equation $\ddot{q}=\omega \left( t,q,\dot{q}\right) $ admits eight Lie
symmetries forming the $sl\left( 3,R\right) $ Lie algebra, then there exists
a transformation that can bring the equation to the form of the free
particle equation $\bar{q}^{\prime \prime }=0$. For further discussion, see
\cite{Leach80a} and references therein.

For third-order ordinary differential equations, various approaches have
been developed to address their linearization. Criteria for linearization
through point and contact symmetries of third-order differential equations
have been established in \cite{mel1, mel2}. Additionally, the Sundman
transformation as a method for linearization has been studied in \cite{sun1,
sun2}. More recently, the Cartan equivalency method was considered in \cite%
{sun3}, where a straightforward procedure was established for the
linearization of third-order differential equations using a four-dimensional
Lie algebra \cite{sun3}.

The linearization of higher-order ordinary differential equations through
Lie symmetries has been extensively discussed in \cite{fam, fam2, fam3} and
related references. However, the linearization of partial differential
equations holds special interest and has yielded many important results, as
seen in \cite{fam4}.

In the context of systems of nonlinear differential equations, the existence
of transformations that linearize these equations is highly significant, as
it offers an alternative approach to studying the integrability of dynamical
systems. Due to the complexity of the problem, various criteria have been
proposed in the literature for the linearization of higher-dimensional
dynamical systems. For example, in \cite{wafo}, it was demonstrated that a
system of second-order differential equations admitting four Lie point
symmetries forming the $A_{4,1}$ or $A_{4,2}$ Lie algebra \cite{patera} can
be transformed into a linear form through a point transformation.
Furthermore, the introduction of complex Lie symmetries has led to new
linearization criteria for systems of second-order ordinary differential
equations, as discussed in the series of studies \cite{ali1, ali2, ali3}.

Nevertheless, the linearization process is inherently a geometric approach.
Therefore, studies in the literature have shown that a system of
differential equations is linearizable based on specific geometric
properties \cite{qa1}. For systems of second-order ordinary differential
equations of the form

\begin{equation*}
\ddot{q}^{i}+\alpha _{srj}^{i}\left( t,q^{k}\right) \dot{q}^{s}\dot{q}^{r}%
\dot{q}^{j}+\beta _{rj}^{i}\left( t,q^{k}\right) \dot{q}^{r}\dot{q}%
^{j}+\gamma _{r}^{i}\left( t,q^{k}\right) \dot{q}^{r}+\delta ^{i}\left(
t,q^{k}\right) =0,
\end{equation*}%
the coefficients $\alpha _{srj}^{i}$,~$\beta _{rj}^{i}$, $\gamma _{r}^{i}$,~$%
\delta ^{i}$ can be related to the connection coefficient of extended
manifold \cite{ami}, where if the connection has zero curvature, that is,
the geometry is flat, then there exist a point transformation which
linearize the latter system, for some applications of this method we refer
the reader in\ \cite{mt1}.

Recently, in \cite{axioms}, a novel approach to geometric linearization for
a family of second-order differential equations was discovered. It was found
that this new family of dynamical systems can be linearized by introducing
new dependent variables. The solution of the extended system can then be
expressed in terms of solutions to linear equations. In a similar vein, it
was discovered that solutions to Einstein's field equations for certain
gravitational models can also be represented in terms of linear equations
\cite{comm}.

In the following section, we establish a new geometric criterion for the
global linearization of a family of constraint Hamiltonian systems.

\section{Constraint Hamiltonian systems}

\label{sec3}

We introduce the Lagrangian function%
\begin{equation}
L\left( N,q^{k},\dot{q}^{k}\right) =\frac{1}{2N}g_{ij}\dot{q}^{i}\dot{q}%
^{j}-NV(q^{k}),  \label{sl.01}
\end{equation}%
where $g_{ij}\left( q^{k}\right) $ is a second-rank tensor with inverse $%
g^{ij}$, and $N\left( t\right) $, $q^{i}\left( t\right) $ are the $\left(
1+n\right) $ degrees of freedom.

For the Lagrangian (\ref{sl.01}) it follows $\det \left\vert \frac{\partial
^{2}L}{\partial Q^{A}\partial Q^{B}}\right\vert =0$, where $Q^{A}=\left(
N\left( t\right) ,q^{i}\left( t\right) \right) $. Hence, the dynamical
system described by the Lagrangian (\ref{sl.01}) is singular.

Variation with respect to the variable $N\left( t\right) $, leads to the
constraint equation $\frac{\partial L}{\partial N}=0$, i.e.%
\begin{equation}
\frac{1}{2N^{2}}g_{ij}\dot{q}^{i}\dot{q}^{j}+V(q^{k})=0.  \label{sl.02}
\end{equation}

We introduce the momentum $p_{i}=\frac{1}{N}g_{ij}\dot{q}^{j}$, thus, the
latter constraint reads%
\begin{equation}
H\left( q,p\right) \equiv \frac{1}{2}g^{ij}p_{i}p_{j}+V(q^{k})=0.
\label{sl.03}
\end{equation}

Function $H\left( q,p\right) $ represents the Hamiltonian of (\ref{sl.01}),
which is constraint due to existence of expression (\ref{sl.03}). Dynamical
systems described by Lagrangians of the form (\ref{sl.01}) are of special
interests in gravitational physics \cite{Rya1} and in other physical
theories (see the discussion in \cite{cont1}). Nevertheless, any regular
Hamiltonian system of the form
\begin{equation}
H_{R}\left( q,p,h\right) \equiv \frac{1}{2}g^{ij}p_{i}p_{j}+U(q^{k})=h,
\label{sl.04}
\end{equation}%
can be written in the singular form by absorbing the integration constant $h$%
, inside the potential, i.e. $U\left( q^{k}\right) =V(q^{k})-h$.

An important characteristic of this family of constrained dynamical systems
is their invariance under time-reparametrization. Below, we demonstrate this
property.

The action integral related to the Lagrangian function (\ref{sl.01}) is%
\begin{equation}
S=\int \left( \frac{1}{2N}g_{ij}\dot{q}^{i}\dot{q}^{j}-NV(q^{k})\right) dt,
\label{sl.05}
\end{equation}%
and the corresponding equations of motion are%
\begin{eqnarray}
\ddot{q}^{i}+\Gamma _{jk}^{i}\dot{q}^{j}\dot{q}^{k}+V^{,i}-2(\ln
N)^{,i}\left( \frac{1}{2}g_{ij}\dot{q}^{i}\dot{q}^{j}+V(q^{k})\right) &=&0,
\label{sl.06} \\
\frac{1}{2N^{2}}g_{ij}\dot{q}^{i}\dot{q}^{j}+V(q^{k}) &=&0,
\end{eqnarray}%
or due to the constraint equation the second-order differential equation
reads
\begin{equation}
\ddot{q}^{i}+\Gamma _{jk}^{i}\dot{q}^{j}\dot{q}^{k}+V^{,i}=0,  \label{sl.07}
\end{equation}
where $\Gamma _{jk}^{i}$ is the Levi-Civita connection related to the tensor
$g_{ij}$.

Under the change of the independent variable $M\left( \tau \right) d\tau =dt$%
, it follows%
\begin{equation}
S=\int \tilde{L}\left( \tilde{N},q^{k},\frac{dq^{k}}{d\tau }\right) d\tau ~,~%
\tilde{N}=N\left( \tau \right) M\left( \tau \right) ,  \label{sl.08}
\end{equation}%
where
\begin{equation}
\tilde{L}\left( \tilde{N},q^{k},\frac{dq^{k}}{d\tau }\right) =\frac{1}{2%
\tilde{N}}g_{ij}\frac{dq^{i}}{d\tau }\frac{dq^{j}}{d\tau }-\tilde{N}%
V(q^{k})~,
\end{equation}%
is the conformally related Lagrangian. The corresponding equations of motion
are derived
\begin{equation}
\frac{d^{2}q^{i}}{d\tau ^{2}}+\Gamma _{jk}^{i}\frac{dq^{j}}{d\tau }\frac{%
dq^{k}}{d\tau }+V^{,i}=0.  \label{sl.09}
\end{equation}%
\begin{equation}
\frac{1}{2N^{2}}g_{ij}\frac{dq^{i}}{d\tau }\frac{dq^{j}}{d\tau }+V(q^{k})=0
\label{sl.10}
\end{equation}%
It is evident that the equations of motion remain invariant under a
time-reparametrization where the parameter$N\left( t\right) $,does not
affect the dynamical behavior of the constrained system. Specifically, the
conformally related Lagrangians $L\left( N,q^{k},\dot{q}^{k}\right) $,~$%
\tilde{L}\left( \tilde{N},q^{k},\frac{dq^{k}}{d\tau }\right) $ have common
equations of motions.

\subsection{Symmetry analysis}

The Lie and Noether symmetries for the constraint Hamiltonian systems with
Lagrangian of the form (\ref{sl.01}) have been investigated in details
before in \cite{ndl}. Owing to the constraint (\ref{sl.02}) the symmetry
analysis differs from that of the regular systems, see for example \cite%
{ansym1}. Indeed for regular dynamical systems and for the equations of
motion of the form (\ref{sl.07}) the Lie symmetries are constructed by the
elements of the projective algebra of the connection $\Gamma _{jk}^{i}$.
Additionally, Noether symmetries are linked to homothetic symmetries of the
metric tensor $g_{ij}$ \cite{ansym1}.

Indeed if $\Phi \left( t,q^{k},p^{k}\right) $ is a conservation law for the
dynamical system with Hamiltonian (\ref{sl.03}), it holds
\begin{equation}
\frac{d\Phi }{dt}\equiv \frac{\partial \Phi }{\partial t}+\{\Phi ,H\}=0,
\end{equation}%
where $\left\{ ,\right\} $ is the Poisson bracket.\

Nevertheless in order to make use of the Hamiltonian constraint the latter
condition can be relaxed such that as follows

\begin{equation}
\frac{d\Phi }{dt}\equiv \frac{\partial \Phi }{\partial t}+\{\Phi ,H\}=\chi
H\approx 0,  \label{sl.11}
\end{equation}%
in which $\chi $ is a conformal factor.

Thus from (\ref{sl.11}) it follows that the symmetries of constraint
Hamiltonian systems are determined by the conformal symmetries of the metric
tensor $g_{ij}$ \cite{ndl}$\,$.

\subsection{Geodesic description}

The process of geometrizing dynamical systems of the form (\ref{sl.01})
involves formulating the dynamical system as a set of geodesic equations,
where the potential term can be interpreted as part of the geometry. Two
different approaches for this are the Jacobi metric and the Eisenhart lift,
both of which incorporate the potential; that is, conservative forces, into
the geometry.

In Jacobi metric\ approach (for a recent discussion see \cite{pr1}), we
introduce $\tilde{N}=\frac{1}{V}$, such that the Lagrangian to be%
\begin{equation}
\bar{L}\left( q^{k},\dot{q}^{k}\right) =\frac{1}{2}\bar{g}_{ij}\dot{q}^{i}%
\dot{q}^{j}-1~,~\bar{g}_{ij}=V\left( q\right) g_{ij}.  \label{sl.12}
\end{equation}%
Thus, the constraint equation (\ref{sl.02}) becomes
\begin{equation}
\frac{1}{2}\bar{g}_{ij}\dot{q}^{i}\dot{q}^{j}+1=0  \label{sl.13}
\end{equation}%
which means that Lagrangian (\ref{sl.12}) describes the time/space -like
geodesic equations for the conformal metric $\bar{g}_{ij}$. \ It is well
known that the geodesic equations of the metric $\bar{g}_{ij}$ which can be
linearized are those of the flat space, that is, $\bar{g}_{ij}$ should be
flat.

The Jacobi metric is obtained through a conformal transformation where $\bar{%
g}{ij}$ is conformally related to the metric tensor $g{ij}$. Conformally
related spacetimes share the same conformal symmetries, implying that the
symmetries and conservation laws of the original dynamical system also exist
for the Jacobi metric. Notably, conformal symmetries in this context
transform into isometries. On the other hand, the Eisenhart lift approach
involves augmenting the dimensionality of the system \cite{el1, el2}.

Geometrization is achieved by introducing additional dimensions through new
dependent variables. The potential term becomes integrated into the newly
extended metric tensor. These extended spaces possess isometries that
correspond to Noetherian conservation laws for the geodesic equations. When
these conservation laws are applied within the extended system, the original
dynamics are recovered.

In the Eisenhart approach, we write the extended Lagrangian function%
\begin{equation}
L_{lift}\left( N,q^{k},\dot{q}^{k},z,\dot{z}\right) =\frac{1}{2N}\left(
g_{ij}\dot{q}^{i}\dot{q}^{j}+\frac{1}{V(q^{k})}\dot{z}^{2}\right) ,
\label{sl.14}
\end{equation}%
with constraint equation and conservation law
\begin{equation}
g_{ij}\dot{q}^{i}\dot{q}^{j}+\frac{1}{V(q^{k})}\dot{z}^{2}=0~,~\frac{1}{%
V\left( q^{k}\right) }\dot{z}=I_{0}.  \label{sl.15}
\end{equation}%
Thus by replacing the $\dot{z}=I_{0}V\left( q^{k}\right) $ in the equations
of motion and the constraint equation, the original system is recovered when
$\left( I_{0}\right) ^{2}=1$. \

Lagrangian function (\ref{sl.14}) describe a sets of null geodesic equations
in the extended manifold with metric
\begin{equation}
\hat{g}_{AB}\left( q^{k},z\right) =g_{ij}\left( q^{k}\right) +\frac{1}{%
V\left( q^{k}\right) }\delta _{A}^{z}\delta _{B}^{z},
\end{equation}%
with $A,B=i,j..,z$, are the indices of the extended geometry.

\section{Geometric Linearization}

\label{sec4}

In this section, we will establish new geometric criteria for the global
linearization of the equations of motion described by the Lagrangian (\ref%
{sl.01}).

We will utilize the two geometrization approaches described previously.
Specifically, for the equivalent Lagrangian (\ref{sl.12}) obtained through
the Jacobi metric approach, we will employ the Eisenhart lift. This allows
us to describe the dynamical system using the equivalent/extended Lagrangian
function
\begin{equation}
\bar{L}_{lift}=\frac{1}{2}\bar{g}_{ij}\dot{q}^{i}\dot{q}^{j}+\frac{1}{2}\dot{%
z}^{2}  \label{sl.16}
\end{equation}%
with constraints
\begin{equation}
\bar{g}_{ij}\dot{q}^{i}\dot{q}^{j}+\dot{z}^{2}=0~,~\dot{z}=\pm 1.
\label{sl.17}
\end{equation}

Lagrangian (\ref{sl.16}) describes the equations of motion for the null
geodesics of the extended metric tensor $\tilde{g}_{AB}\left( q^{k},z\right)
=\bar{g}_{ij}\left( q^{k}\right) +\delta _{A}^{z}\delta _{B}^{z}$.

Consequently, if the metric tensor $\tilde{g}_{AB}$ is conformally flat then
there exist a point transformation $\left( q^{A},z\right) =Q^{A}$ such that $%
\tilde{g}_{AB}=e^{2U}\eta _{AB}$, where $\eta _{AB}$ is the flat space in
diagonal coordinates. Thus, Lagrangian (\ref{sl.16}) and the constraint (\ref%
{sl.17}) becomes%
\begin{equation}
\bar{L}_{lift}=e^{2U}\eta _{AB}\dot{Q}^{A}\dot{Q}^{B},  \label{sl.18}
\end{equation}%
\begin{equation}
\eta _{AB}\dot{Q}^{A}\dot{Q}^{B}=0.  \label{sl.19}
\end{equation}%
Therefore after the change of the independent variable $d\tau =e^{-2U}dt$,
we end with the equivalent geodesic equations%
\begin{equation}
\ddot{Q}^{A}=0~,  \label{sl.20}
\end{equation}%
which are the equations of motion for the free particle in the flat space.

Thus, regarding the linearization of a dynamical system described by the
singular Lagrangian function (\ref{sl.01}) the following theorem states.

\begin{theorem}
\label{Theo1}The $n$-dimensional, \thinspace with $n\geq 2\,,$ constraint
Hamiltonian system described by the Lagrangian function (\ref{sl.01}) is
globally linearizable if one of the following equivalent statements are true

A) The admitted (non trivial) Noether symmetries of the constraint
Hamiltonian system are $\frac{n\left( n+1\right) }{2}$.

B) The Jacobi metric $\bar{g}_{ij}$ is maximally symmetric.

C)\ The extended $1+n$ decomposable space with metric $\tilde{g}_{AB}$ is
conformally flat.
\end{theorem}

\begin{proof}
All statements of Theorem \ref{Theo1} are equivalent. As we discussed before
if the extended $1+n$ decomposable space with metric $\tilde{g}_{AB}$ is
conformally flat then the equations of motions are linearizable through
point transformations in the space $\left\{ q^{A},z\right\} $.

Nevertheless, $\tilde{g}_{AB}$ is conformally flat if and only if the $n$%
-dimensional space with metric $\bar{g}_{ij}$ is maximally symmetric (see
Proposition two in \cite{mtgrg}).

However, when $\bar{g}_{ij}$ is maximally symmetric it means $\bar{g}_{ij}$
admits $\frac{n\left( n+1\right) }{2}$ isometries, as many as the Noether
symmetries for the original Lagrangian (\ref{sl.01}).

The inverse proof is straightforward.
\end{proof}

It is important to note that two-dimensional spaces are maximally symmetric
when they have constant curvature. Conversely, three-dimensional spaces are
conformally flat when the Cotton-York tensor vanishes. Therefore, for the
case of two-dimensional systems, from Theorem \ref{Theo1}, we derive the
following corollary.

\begin{corollary}
\label{corol1}The two-dimensional constraint Hamiltonian systems described
by the Lagrangian function (\ref{sl.01}) are globally linearizable if one of
the following equivalent states are true

A) The admitted Noether symmetries of the constraint Hamiltonian system are
three.

B) There exist a point transformation where the two-dimensional Jacobi
metric can be written as $\bar{g}_{ij}=diag\left( e^{2U\left( x,y\right)
},e^{2U\left( x,y\right) }\right) $, and $U\left( x,y\right) $ is a solution
of the equation $U_{,xx}+U_{,yy}+2\kappa e^{2U}=0$, in which $\kappa $ is
the curvature of the two-dimensional space and $V\left( x,y\right)
=e^{-2U\left( x,y\right) }$

C) The Cotton-York tensor $C_{\mu \nu \kappa }=R_{\mu \nu ;\kappa
}-R_{\kappa \nu ;\mu }+\frac{1}{4}\left( R_{;\nu }g_{\mu \kappa }-R_{;\kappa
}g_{\mu \nu }\right) $ for the three-dimensional extended metric $\tilde{g}%
_{AB}$ has zero components.
\end{corollary}

\begin{corollary}
\label{corol2}The one-dimensional constraint Hamiltonian system described by
the Lagrangian function (\ref{sl.01}) is always globally linearizable for
arbitrary potential function.
\end{corollary}

\begin{proof}
Consider the one-dimensional constraint Hamiltonian system with Lagrangian
function $L\left( N,q,\dot{q}\right) =\frac{1}{2N}\dot{q}^{2}-NV\left(
q\right) $, then after the application of Jacobi's and Eisenhart's methods
we find the equivalent dynamical system of null geodesics for the
two-dimensional space. All two-dimensional spaces are conformally flat which
means that the null geodesics can be globally linearizable. In particular
the extended two-dimensional space is described by the decomposable line
element $ds^{2}=\frac{1}{V\left( q\right) }dq^{2}+dz^{2}$. $\ $Thus, under
the change of variable $\sqrt{\frac{1}{V\left( q\right) }}dq=dY$, it follows
$ds^{2}=dY^{2}+dz^{2}$, which is the flat space.

On the other hand, the Jacobi metric approach leads to the equivalent system
of time/space -like geodesic equations for the one-dimensional line element $%
d\bar{s}^{2}=\frac{1}{V\left( q\right) }dq^{2}$. Because the one-dimensional
space is the flat space the point transformation $\sqrt{\frac{1}{V\left(
q\right) }}dq=dY$ leads to the linearization of the dynamical system.
\end{proof}

\section{Applications}

\label{sec5}

In the following lines, we demonstrate the application of Theorem \ref{Theo1}
for the geometric linearization of Hamiltonian systems of special interests.

\subsection{Exponential interaction}

Consider two particles which interact with an exponential law, similar to
that of the Toda lattice. The Lagrangian which describes the dynamical
system is%
\begin{equation}
L\left( q_{1},\dot{q}_{1},q_{2},\dot{q}_{2}\right) =\frac{1}{2}\left( \dot{q}%
_{1}^{2}+\dot{q}_{2}^{2}\right) -V_{0}e^{2\gamma \left( q_{1}-q_{2}\right) }.
\end{equation}%
where the Hamiltonian function is $H=h$.\newline

We introduce the equivalent singular Lagrangian%
\begin{equation}
\bar{L}\left( N,q_{1},\dot{q}_{1},q_{2},\dot{q}_{2}\right) =\frac{1}{2N}%
\left( \dot{q}_{1}^{2}+\dot{q}_{2}^{2}\right) -N\left(
V_{0}e^{q_{1}-q_{2}}-h\right) ,
\end{equation}%
with constraint equation%
\begin{equation}
\bar{H}\equiv \frac{1}{2}\left( p_{1}^{2}+p_{2}^{2}\right)
+V_{0}e^{q_{1}-q_{2}}-h=0,
\end{equation}%
in which $p_{1}=\dot{q}_{1}$,~$p_{2}=\dot{q}_{2}$.

We employ the Eisenhart lift and we write the equivalent Lagrangian of the
form (\ref{sl.14}), that is,
\begin{equation}
\bar{L}_{lift}\left( N,q_{1},\dot{q}_{1},q_{2},\dot{q}_{2}\right) =\frac{1}{%
2N}\left( \left( \dot{q}_{1}^{2}+\dot{q}_{2}^{2}\right) +\frac{1}{\left(
V_{0}e^{q_{1}-q_{2}}-h\right) }\dot{z}^{2}\right) ,
\end{equation}%
where the extended three-dimensional metric reads%
\begin{equation}
ds_{lift}^{2}=dq_{1}^{2}+dq_{2}^{2}+\frac{1}{\left(
V_{0}e^{q_{1}-q_{2}}-h\right) }dz^{2}.
\end{equation}%
We calculate the Cotton-York tensor and we found that it is zero when $h=0$.
Thus, from Corollary \ref{corol1} it follows that the dynamical system can
be written in the equivalent form of the three-dimensional free particle of
the flat geometry.

Under the he conformal transformation $d\bar{s}_{lift}^{2}=\frac{1}{XY}%
ds_{lift}^{2}~$and the change of variables%
\begin{eqnarray}
ix &=&-\frac{1}{4}\left( 1+i\right) \ln Y+\frac{1}{4}\left( 1-i\right) \ln X,
\\
iy &=&\frac{1}{4}\left( 1-i\right) \left( i\ln X-\ln Y\right) ,
\end{eqnarray}%
the extended space $d\bar{s}_{lift}^{2}$ becomes%
\begin{equation}
ds_{lift}^{2}=\frac{1}{2}dXdY+dz^{2}.
\end{equation}%
where the null geodesics read%
\begin{equation}
\ddot{X}=0~,\ddot{Y}=0\text{ },~\ddot{z}=0,
\end{equation}%
and constraint equation%
\begin{equation}
\frac{1}{2}\dot{X}\dot{Y}+\dot{z}^{2}=0\text{.}
\end{equation}

\subsection{Two-dimensional oscillator with corrections}

We assume the following singular Lagrangian, which describes a
two-dimensional oscillator with correction terms, that is,
\begin{equation}
L\left( N,x,\dot{x},y,\dot{y}\right) =\frac{1}{2N}\left( \dot{x}^{2}+\dot{y}%
^{2}\right) +N\left( 1+\frac{\kappa }{4}\left( x^{2}+y^{2}\right) \right)
^{2}.  \label{ap.11}
\end{equation}

We calculate the Jacobi metric, that is, the corresponding line element is
of the form
\begin{equation}
ds_{Jacobi}^{2}=\frac{1}{\left( 1+\frac{\kappa }{4}\left( x^{2}+y^{2}\right)
\right) ^{2}}\left( dx^{2}+dy^{2}\right)  \label{ap.12}
\end{equation}%
from where we observe that is a space of constant $\kappa $ nonzero
curvature, that it, it is a maximally symmetric space.

We employ Eisenhart's lift and we write the equivalent extended Lagrangian
\begin{equation}
L_{lift}\left( N,x,\dot{x},y,\dot{y}\right) =\frac{1}{2}\frac{\dot{x}^{2}+%
\dot{y}^{2}}{\left( 1+\frac{\kappa }{4}\left( x^{2}+y^{2}\right) \right) ^{2}%
}+\frac{1}{2}\dot{z}^{2},  \label{ap.13}
\end{equation}%
with constraint
\begin{equation}
\frac{1}{2}\frac{\dot{x}^{2}+\dot{y}^{2}}{\left( 1+\frac{\kappa }{4}\left(
x^{2}+y^{2}\right) \right) ^{2}}+\frac{1}{2}\dot{z}^{2}=0.  \label{ap.14}
\end{equation}

For the three-dimensional extended space with line element
\begin{equation}
ds_{Lift}^{2}=ds_{Jacobi}^{2}+dz^{2}  \label{ap.15}
\end{equation}%
we calculate the Cotton-York tensor components which are zero. Hence, the
line element (\ref{ap.15}) is conformally flat.

Under the change of variables
\begin{equation}
x=\frac{X\left( Z+\sqrt{Z^{2}+16\kappa \left( X^{2}+Y^{2}\right) }\right) }{%
2\kappa \left( X^{2}+Y^{2}\right) }~,~y=\frac{Y\left( Z+\sqrt{Z^{2}+16\kappa
\left( X^{2}+Y^{2}\right) }\right) }{2\kappa \left( X^{2}+Y^{2}\right) },
\label{ap.16}
\end{equation}%
\begin{equation}
e^{\sqrt{K}z}=2\kappa \frac{\left( X^{2}+Y^{2}\right) }{Z+\sqrt{%
Z^{2}+16\kappa \left( X^{2}+Y^{2}\right) }}\left( \frac{8\left(
X^{2}+Y^{2}\right) }{\left( Z+\sqrt{Z^{2}+16\kappa \left( X^{2}+Y^{2}\right)
}\right) }+Z\right) ,  \label{ap.17}
\end{equation}%
the three-dimensional space (\ref{ap.15}) \ reads%
\begin{equation}
ds_{Lift}^{2}=e^{\sqrt{K}z\left( X,Y,Z\right) }\left(
dX^{2}+dY^{2}+dZ^{2}\right) ,  \label{ap.18}
\end{equation}%
with null geodesic equations%
\begin{equation}
\ddot{X}=0~,~\ddot{Y}=0~,\text{ }\ddot{Z}=0,  \label{ap.19}
\end{equation}%
and constraint $\dot{X}^{2}+\dot{Y}^{2}+\dot{Z}^{2}=0$. \

This two-dimensional linearizable example can be generalized to $n-$%
dimensional system in a similar way.

\begin{corollary}
\label{corol3}The $n-$dimensional dynamical system with Lagrangian
\begin{equation}
L\left( N,q^{k},\dot{q}^{k}\right) =\frac{1}{2N}\left( \eta _{ij}\dot{q}^{i}%
\dot{q}^{j}\right) +N\left( 1+\frac{\kappa }{4}\left( \eta
_{ij}q^{i}q^{j}\right) \right) ^{2},
\end{equation}%
can be linearized though the Jacobi metric and the Eisenhart lift.
\end{corollary}

\subsection{The Szekeres system}

Consider the two-dimensional dynamical system described by the regular
Lagrangian function \cite{anszek,qq}%
\begin{equation}
L_{R}\left( u,\dot{u},v,\dot{v}\right) =\dot{u}\dot{v}-\frac{v}{u^{2}},
\label{ap.01}
\end{equation}%
and equations of motion%
\begin{equation}
\ddot{u}+u^{-2}=0~,~\ddot{v}-2vu^{-3}=0.  \label{ap.02}
\end{equation}%
Furthermore, the Hamiltonian function is
\begin{equation}
H\equiv \dot{u}\dot{v}+\frac{v}{u^{2}}=h  \label{ap.03}
\end{equation}%
where $h$ is a constant.

The system described above corresponds to the dynamics governed by the
Einstein field equations for Szekeres spacetimes, commonly referred to as
the Szekeres system \cite{szek0}. Szekeres spacetimes represent exact
inhomogeneous solutions with a dust fluid and find various applications in
gravitational physics. For more detailed information, the interested reader
is referred to \cite{andrez}.

We introduce the equivalent singular Lagrangian function%
\begin{equation}
L\left( N,u,\dot{u},v,\dot{v}\right) =\frac{\dot{u}\dot{v}}{N}-N\left( \frac{%
v}{u^{2}}-h\right) ,  \label{ap.04}
\end{equation}%
such that to write the original system in the form of a constraint
Hamiltonian dynamical system.

The latter singular Lagrangian for arbitrary parameter $h$ has not any
Noetherian symmetry. In particular it admits a hidden symmetry related to a
quadratic conservation law \cite{anszek}. However, for $h=0$, Lagrangian (%
\ref{ap.04}) admits three Noether symmetries%
\begin{equation}
X_{1}=\frac{1}{v}\partial _{v}~,~X_{2}=u^{2}\partial _{u}~,~X_{3}=v\partial
_{v}+2u\partial _{u}.  \label{ap.05}
\end{equation}%
Thus, according to the first statement of Theorem \ref{Theo1}, on the
surface with $h=0$, the equations of motions can be linearized. Moreover,
the solution space with initial conditions $h=0$, is equivalent to the
equations of motion for the two-dimensional flat space.

Indeed, we employ Jacobi's approach and from (\ref{ap.04}) we define the
equivalent geodesic Lagrangian%
\begin{equation}
\bar{L}\left( u,\dot{u},v,\dot{v}\right) =\frac{v}{u^{2}}\dot{u}\dot{v}-1.
\label{ap.06}
\end{equation}%
Hence, under the change of variables
\begin{equation}
\frac{du}{u^{2}}=dU,~vdv=dV,  \label{ap.06a}
\end{equation}%
it follow%
\begin{equation}
L\left( U,\dot{U},V,\dot{V}\right) =\dot{U}\dot{V}-1,  \label{ap.07}
\end{equation}%
with equations of motion%
\begin{equation}
\ddot{U}=0,~\ddot{V}=0
\end{equation}%
and constraint equation
\begin{equation}
\dot{U}\dot{V}+1=0.
\end{equation}

\subsection{The Cosmological Constant in Szekeres model}

The introduction of the cosmological constant in the Szekeres model \cite%
{ccbar} leads to the modification of the regular Lagrangian (\ref{ap.01}).

Specifically, the new dynamics follows from the Lagrangian function \cite%
{zzbar}
\begin{equation}
L_{R}^{\Lambda }\left( u,\dot{u},v,\dot{v}\right) =\dot{u}\dot{v}-\left(
\frac{v}{u^{2}}-\Lambda uv\right) ,  \label{ap.08}
\end{equation}%
in which $\Lambda $ is the cosmological constant term.

In a similar approach with before we introduce the singular Lagrangian
function%
\begin{equation}
L^{\Lambda }\left( N,u,\dot{u},v,\dot{v}\right) =\dot{u}\dot{v}-\left( \frac{%
v}{u^{2}}-\Lambda uv-h\right)  \label{ap.09}
\end{equation}

Thus, the Jacobi metric is defined as%
\begin{equation}
ds_{Jacobi}^{2}=\frac{1}{\left( \frac{v}{u^{2}}-\Lambda uv-h\right) }dudv.
\label{ap.10}
\end{equation}%
We calculate the Ricciscalar for the two-dimensional space, it is,
\begin{equation}
R_{Jacobi}=-\frac{4\left( 2+\Lambda u^{3}\right) h}{u\left( v\left( \Lambda
u^{3}-1\right) +hu^{2}\right) }.
\end{equation}
Hence, for $h=0$, the two-dimensional Jacobi metric (\ref{ap.10}) describes
the flat space, that is, the equations of motion can be linearized though a
point transformation.

Indeed, under the point transformation
\begin{equation}
\frac{du}{\left( \frac{1}{u^{2}}-\Lambda u\right) }=dU~,~\frac{dv}{v}=dV,
\end{equation}
the Jacobi metric reads
\begin{equation}
ds_{Jacobi}^{2}=dUdV,
\end{equation}%
which leads to the geodesic equations of the flat space, that is,
\begin{equation}
\ddot{U}=0~,~\ddot{V}=0.
\end{equation}

\subsection{Static spherical symmetric spacetime with charge}

Einstein's gravitational field equations for the a static spherical
symmetric spacetime with a charge are described by the variation of the
singular Lagrangian \cite{dim1}%
\begin{equation}
L^{RN}\left( N,a,a^{\prime },b,b^{\prime },\zeta ,\zeta ^{\prime }\right) =%
\frac{1}{2N}\left( 8ba^{\prime }b^{\prime }+4ab^{\prime 2}+4\frac{b^{2}}{a}%
\zeta ^{\prime 2}\right) +2Na,  \label{rn.01}
\end{equation}%
where $N,a,b$ are the scale factors of the background geometry with line
element%
\begin{equation}
ds^{2}=-a\left( r\right) ^{2}dt^{2}+N\left( r\right) ^{2}dr^{2}+b^{2}\left(
r\right) \left( d\theta ^{2}+\sin ^{2}\theta d\phi ^{2}\right) ,
\end{equation}%
and a prime means total derivative with respect the independent parameter $r$%
, i.e. $a^{\prime }=\frac{da}{dr}$. \ Function $\zeta $ is related to the
charge. The solution of the field equations is known as the Reissner-Nordstr%
\"{o}m black hole \cite{sch3,sch4}

The kinetic term of the singular Lagrangian (\ref{rn.01}) is defined by a
three-dimensional space. The admitted Noether symmetries of Lagrangian (\ref%
{rn.01}) are calculated to be six \cite{dim1}, they are \cite{comm}
\begin{equation*}
X^{1}=\frac{1}{ab}\partial _{a}~,~X^{2}=-a\partial _{a}+b\partial
_{b}-z\partial _{\zeta }~,~
\end{equation*}%
\begin{equation*}
X^{3}=-\left( \frac{a}{2b}+\frac{z^{2}}{ab}\right) \partial _{a}+\partial
_{b}-\frac{\zeta }{b}\partial _{\zeta }~,
\end{equation*}%
\begin{equation*}
X^{4}=-a\zeta \partial _{a}+b\zeta \partial _{b}+\left( \frac{a^{2}}{4}-%
\frac{z^{2}}{2}\right) \partial _{\zeta },
\end{equation*}%
\begin{equation*}
X^{5}=\frac{2\zeta }{ab}\partial _{a}+\frac{1}{b}\partial _{\zeta
}~,~X^{6}=\partial _{\zeta }.
\end{equation*}

Therefore, case A of Theorem \ref{Theo1} states the field equations can be
written in the equivalent form of the free particle. This transformation
derived before in \cite{comm} where the common solution space for a large
families of gravitational models investigated.

We employ the Eisenhart lift and we write the extended line element%
\begin{equation}
ds_{lift}^{2~m}=\frac{1}{N}\left( 8b~da~db+4a~db^{2}+4\frac{b^{2}}{a}d\zeta
^{2}-\frac{d\psi ^{2}}{2a}\right)
\end{equation}%
and under the point transformation
\begin{equation}
a=\sqrt{\frac{X+Y}{X-Y}+\frac{z^{2}}{\left( X-Y\right) ^{2}}},~~\zeta =\frac{%
z}{\left( X-Y\right) ^{2}},
\end{equation}%
the latter line element becomes%
\begin{equation}
ds^{RN~2}=\frac{1}{n}\frac{X-Y}{\sqrt{X^{2}-Y^{2}+z^{2}}}\left(
4dX^{2}-4dY^{2}+4dz^{2}-d\psi ^{2}\right)
\end{equation}%
where the corresponding null geodesics are written in the linear form
\begin{equation}
X^{\prime \prime }=0~,~Y^{\prime \prime }=0,~z^{\prime \prime }=0~,~\psi
^{\prime \prime }=0.
\end{equation}%
and Hamiltonian constraint%
\begin{equation}
\dot{X}^{2}-\dot{Y}^{2}+\dot{z}^{2}-\frac{1}{4}\dot{\psi}^{2}=0.
\end{equation}

\section{Conclusions}

\label{sec6}

In this study, by using the Jacobi metric and the Eisenhart lift, we
established a new criterion for the global linearization of constrained
Hamiltonian systems. The requirements for the dynamical system are to be in
the form of (\ref{sl.01}) with the constraint expression (\ref{sl.02}). The $%
n$-dimensional dynamical system must admit $\frac{n\left( n+1\right) }{2}$
Noether point symmetries, which correspond to a number of $\frac{n\left(
n+1\right) }{2}$ independent conservation laws. This property indicates that
the given dynamical system, constrained by equation (\ref{sl.02}), posses
the property of superintegrability.

The linearization of the dynamical systems is achieved through geometry and
is based on the linearization of the equivalent system that describes
geodesic equations for an extended geometry. The main result of this
analysis is summarized in Theorem \ref{Theo1}, where three equivalent
statements related the number of admitted symmetries, and the geometric
characteristics of the Jacobi metric and of the extended Eisenhart metric
are given.

Two immediate results are given in Corollaries \ref{corol1} and \ref{corol2}%
. Corollary \ref{corol1} specialize the statements of Theorem \ref{Theo1}
for the case of two-dimensional dynamical systems, while Corollary \ref%
{corol2} states that all one-dimensional constraint Hamiltonian systems can
be globally linearized. This geometric linearization for the one-dimensional
systems is possible, because the two-dimensional extended Eisenhart space is
always conformally flat.

To demonstrate the application of this new geometric approach, we presented
a series of applications of special interest. We focused on two dynamical
systems of analytic mechanics and on some gravitational models. Specifically
we consider the Szekeres system, which describes inhomogeneous cosmological
models with or without the cosmological constant term, and the
Reissner-Nordstr\"{o}m black hole.

A criticism that can be made is that this specific algorithm fails in the
case of the simplest maximally symmetric system, which is that of the $n$%
-dimensional harmonic oscillator with Lagrangian
\begin{equation}
L\left( N,q^{k},\dot{q}^{k}\right) =\frac{1}{2}\eta _{ij}\dot{q}^{i}\dot{q}%
^{j}-\frac{\omega ^{2}}{2}\eta _{ij}q^{i}q^{j}.  \label{d.01}
\end{equation}

We should make it clear that the geometric linearization discussed in this
study is based on the existence of symmetries generated by conformal
symmetries of the metric tensor $g_{ij}$, which defines the kinetic term in (%
\ref{sl.01}). As far as the symmetries of the oscillator (\ref{d.01}) are
concerned, they are related to the projective algebra \cite{anpo} and not to
the conformal algebra. Nevertheless, only the isometries are the common
subalgebra for the projective and conformal algebras. An alternative
Eisenhart lift has been proposed in \cite{axioms}, where the oscillator is
reduced to the free particle. This is possible should the introduction of an
extended space which belongs to the pp-wave geometries. In this case the
dimension of the extended space is increased by two and the resulting
Eisenhart metric is conformally flat.

Thus, any dynamical system of the form (\ref{sl.01}) where an extended
Eisenhart metric can be constructed to be conformally flat, can be
linearized.\ Which means that Theorem \ref{Theo1} can be generalized for
other family of Eisenhart lifts. \

This work opens new directions in the study of nonlinear differential
equations and demonstrates that geometry is a powerful tool for the study of
dynamical systems.

A natural extension of this geometric consideration is for the study of the
Klein-Gordon equation which describes the quantum limit of the constraint
Hamiltonian system with Lagrangian (\ref{sl.01}).

Specifically, the quantization of the constraint Hamiltonian (\ref{sl.03})
leads to the Yamabe equation \cite{yam1}%
\begin{equation}
\hat{\Delta}\Psi +V\left( q^{k}\right) \Psi =0,  \label{d.02}
\end{equation}%
where
\begin{equation}
\hat{\Delta}=\Delta +\frac{n-2}{4\left( n-1\right) }R  \label{d.03}
\end{equation}%
and $\Delta =\frac{1}{\sqrt{\left\vert g\right\vert }}\frac{\partial }{%
\partial q^{i}}\left( \sqrt{g}g^{ij}\frac{\partial }{\partial x^{j}}\right) $%
\ is the usual Laplace operator and $R$\ is the Ricci scalar for the metric
tensor $g_{ij}$\ which defines the kinetic term and $n=\dim g_{ij}$. The
introduction of the second term in (\ref{d.03}) is necessary in order
equation the equation to be invariant under conformal transformations.

We demonstrate the application of the geometric approach in equation (\ref%
{d.02}) let us study the Yamabe equation for the Szekeres system (\ref{ap.04}%
) with $h=0$. The equation which describes the quantum system is%
\begin{equation}
\frac{u^{2}}{v}\left( \Psi _{,uv}-\frac{v}{u^{2}}\Psi \right) =0,~\Psi =\Phi
\left( u,v\right)  \label{d.04}
\end{equation}

Assume now equation
\begin{equation}
\frac{u^{2}}{v}\Phi _{,uv}+\Phi _{,zz}=0~,~\Phi =\Phi \left( u,v,z\right) .
\label{d.05}
\end{equation}%
The vector field $\partial _{z}-i\Phi \partial _{\Phi }$, is a Lie symmetry
for equation (\ref{d.05}). The corresponding invariants are $\left\{
u,v,\Phi =\Psi e^{iz}\right\} $. Thus by replacing in (\ref{d.05}) we end
with equation (\ref{d.04}).

Under the change of variables\ (\ref{ap.06a}), and $U=X+Y~,~V=X-Y$, equation
(\ref{d.05}) reads%
\begin{equation}
\Phi _{,XX}-\Phi _{,YY}+\Phi _{,zz}=0,  \label{d.06}
\end{equation}%
which is the wave equation for the three-dimensional flat space. Equation (%
\ref{d.06}) is maximally symmetric and admits ten Lie point symmetries plus
the infinity symmetries related to the infinity number of solutions of the
linear equation (for more details we refer to \cite{yam2}). On the other
hand, equation (\ref{d.04}) admits only three Lie point symmetries (plus the
infinity symmetries).Therefore, the symmetries of the maximally symmetric
equation (\ref{d.05}) which does not survive under the reduction with the
invariants $\left\{ u,v,\Phi =\Psi e^{iz}\right\} $, becomes nonlocal
symmetries, which can be used for the construction of new solutions for the
inhomogeneous equation (\ref{d.04}) related to solutions for the homogeneous
equation (\ref{d.06})

In a future work, we plan to investigate further this geometric
consideration for the analysis of other dynamical systems and of partial
differential equations.

\bigskip

\textbf{Data Availability Statements:} Data sharing is not applicable to
this article as no datasets were generated or analyzed during the current
study.

\bigskip

\textbf{Acknowledgements: }The author thanks the support of Vicerrector\'{\i}%
a de Investigaci\'{o}n y Desarrollo Tecnol\'{o}gico (Vridt) at Universidad
Cat\'{o}lica del Norte through N\'{u}cleo de Investigaci\'{o}n Geometr\'{\i}%
a Diferencial y Aplicaciones, Resoluci\'{o}n Vridt No - 096/2022 and Resoluci%
\'{o}n Vridt No - 098/2022.


\bigskip

\begin{thebibliography}{99}
\bibitem{lie1} S. Lie,Theorie der Transformationsgrupprn: Vol I, Chelsea,
New York (1970)

\bibitem{lie2} S. Lie, Theorie der Transformationsgrupprn: Vol II, Chelsea,
New York (1970)

\bibitem{lie3} S. Lie,Theorie der Transformationsgrupprn: Vol III, Chelsea,
New York (1970)

\bibitem{ibra} N.H. Ibragimov, CRC Handbook of Lie Group Analysis of
Differential Equations, Volume I: Symmetries, Exact Solutions, and
Conservation Laws, CRS Press LLC, Florida (2000)

\bibitem{Stephani} H. Stephani, Differential Equations: Their Solutions
Using Symmetry, Cambridge University Press, New York, (1989)

\bibitem{olver} P.J. Olver, Applications of Lie Groups to Differential
Equations, Springer-Verlag, New York, (1993)

\bibitem{Bluman} G.W. Bluman and S. Kumei, Symmetries of Differential
Equations, Springer-Verlag, New York, (1989)

\bibitem{sw7} S. Meleshko and N. Samatova, Group classification of the
two-dimensional shallow water equations with the beta-plane approximation of
coriolis parameter in Lagrangian coordinates, Comm. Nonl. Sci. Num. Sim. 90,
105337 (2020)

\bibitem{sw3} A.A. Chesnokov, Symmetries and exact solutions of the rotating
shallow-water equations, J. Appl. Mech. Techn. Phys. 49, 737 (2008)

\bibitem{hl0} M. Torrisi and R.\ Tracina, Symmetries and Conservation Laws
for a Class of Fourth-Order Reaction--Diffusion--Advection Equations,
Symmetry 15, 1936 (2023)

\bibitem{hl0a} L.\ Blanco Diaz, C. Sardon, F.J. Alburquerque and J. de
Lucas, Geometric Numerical Methods for Lie Systems and Their Application in
Optimal Control, Symmetry 15, 1285 (2023)

\bibitem{sw15} V.A.Dorodnitsyn and E.I. Kaptsov, Shallow water equations in
Lagrangian coordinates: Symmetries, conservation laws and its preservation
in difference models, Comm. Nonl. Sci. Num. Sim. 89, 105343 (2020)

\bibitem{ami2} A.V. Aminova, Projective transformations and symmetries of
differential equations, Sbornik Mathematics 186 (12), 1711 (1995)

\bibitem{hl1} H. Liu and Y. Yun, Lie Symmetry Analysis and Conservation Laws
of Fractional Benjamin--Ono Equation, Symmetry 16, 473 (2024)

\bibitem{b2} A.H. Bokhari, A. H. Kara, A. R. Kashif, and F.D. Zaman, Noether
Symmetries Versus Killing Vectors and Isometries of Spacetimes,
International Journal of Theoretical Physics, 45, 1029 (2006)

\bibitem{b3} A.H. Bokhari, Johnpillai, A.G., Kara, A.H., Mahomed, F.M., and
F.D. Zaman, Classification of Static Spherically Symmetric Spacetimes by
Noether Symmetries, International Journal of Theoretical Physics 52, 3534,
(2013)

\bibitem{oli} F. Oliveri, Lie Symmetries of Differential Equations:
Classical Results and Recent Contributions, Symmetry 2, 658 (2010)

\bibitem{prd} J. Belmonte-Beitia, V.M. P\'{e}rez-Garc\'{\i}a, Vadym
Vekslerchik, and Pedro J. Torres, Lie Symmetries and Solitons in Nonlinear
Systems with Spatially Inhomogeneous Nonlinearities, Phys. Rev. Lett. 98,
064102 (2007)

\bibitem{sja} S. Jamal, Quadratic integrals of a multi-scalar cosmological
model, Mod. Phys.\ Lett. A. 35, 2050068 (2020)

\bibitem{sa1} S.A.\ Hojman, A new conservation law constructed without using
either Lagrangians or Hamiltonians, J. Phys. A.: Math. Gen. 25, L291 (1992)

\bibitem{sa2} T. Pillay and P.G.L. Leach, Comment on a theorem of Hojman and
its generalizations, J. Phys. A.: Math. Gen. 29, 6999 (1996)

\bibitem{sa3} H.-B. Zhang and L.-Q. Chen, The Unified Form of Hojman's
Conservation Law and Lutzky's Conservation Law, J. Phys. Soc. Jpn. 74, 905
(2005)

\bibitem{jm0} S. Chanda, G.W. Gibbons and P.\ Guha, Jacobi--Maupertuis
metric and Kepler equation, Int. J. Geom. Meth. Mod. Phys. 14, 1730002 (2017)

\bibitem{jm1} M. Szyd\l owski, Desingularization of Jacobi metrics and chaos
in general relativity, J. Math. Phys. 40, 3519 (1999)

\bibitem{jm2} A. Duenas-Vidal and O.L. Andino, The Jacobi metric approach
for dynamical wormholes, Gen. Relat. Grav. 55, 9 (2023)

\bibitem{jm3} M. Szydlowski, M. Heller and W. Sasin, Geometry of spaces with
the Jacobi metric, J. Math. Phys. 37, 346 (1996)

\bibitem{jm4} G.W. Gibbons, The Jacobi metric for timelike geodesics in
static spacetimes, Cl.ass. Quantum Grav. 33, 025004 (2016)

\bibitem{pr1} P. Maraner, On the Jacobi metric for a general Lagrangian
system, J. Math. Phys. 60, 112901 (2019)

\bibitem{el1} L.P. Eisenhart, Dynamical Trajectories and Geodesics, Annals.
Math. 30, 591 (1928)

\bibitem{duval} C. Duval, G. Burdet, H.P. Kunzle and M. Perrin, Bargmann
structures and Newton-Cartan theory, Phys. Rev. D 1841, 31 (1985)

\bibitem{kk1} T. Kaluza, Zum unit\"{a}tsproblem der physik Itzungsber,
Preuss. Akad. Wiss. Berlin 966 (1921)

\bibitem{kk2} O. Klein, The atomicity of electricity as a quantum theory
law, Nature 118, 516 (1926)

\bibitem{lf1} A.P. Fordy and A. Galajinsky, Eisenhart lift of 2-dimensional
mechanics, Eur. Phys. J.\ C 79, 301 (2019)

\bibitem{lf1a} M. Cari\~{n}ena, F. Jose, F.J. Herranz, M. Fern\'{a}ndez-Ra%
\~{n}ada Men\'{e}ndez De Luarca, Superintegrable systems on 3-dimensional
curved spaces: Eisenhart formalism and separability, J. Math. Phys. 58,
022701 (2017)

\bibitem{lf1b} M. Cariglia and G.W. Gibbons, Generalised Eisenhart lift of
the Toda chain, J. Math. Phys. 55, 022701 (2014)

\bibitem{lf1c} H.\ Zhang and Q.-Y. Hao, Eisenhart lift for Euler's problem
of two fixed centers, Appl. Math. Compt. 350, 305 (2019)

\bibitem{lf3} K. Fin, S. Karamitsos and A. Pilaftsis, Quantizing the
Eisenhart lift, Phys. Rev.\ D 103, 065004 (2021)

\bibitem{lf4} K. Fin, The Eisenhart Lift. In: Geometric Approaches to
Quantum Field Theory, Springer Theses, Springer, Cham (2021)

\bibitem{lf5} M. Cariglia, Hidden symmetries of Eisenhart-Duval lift metrics
and the Dirac equation with flux, Phys. Rev. D 86, 084050 (2012)

\bibitem{lf6} M. Cariglia, A. Galajinsky, G.W. Gibbons and P.A. Horvathy,
Cosmological aspects of the Eisenhart--Duval lift, Eur. Phys. J. C 78, 314
(2018)

\bibitem{lf7} A. Balcerzak and M. Lisaj, Spinor wave function of the
Universe in non-minimally coupled varying constants cosmologies, Eur. Phys.
J. C 83, 401 (2023)

\bibitem{lf8} C. Duval, G.W. Gibbons and P.A. Horvathy, Conformal and
projective symmetries in Newtonian cosmology, J. Geom. Phys. 112, 197 (2017)

\bibitem{lf9} N. Kan, T. Aoyama, T. Hasegawa and K. Shiraishi,
Eisenhart-Duval lift for minisuperspace quantum cosmology, Phys. Rev. D 104,
086001 (2021)

\bibitem{lf10} A. Paliathanasis, Classical and quantum solutions in scalar
field cosmology via the Eisenhart lift and linearization, Phys. Dark Univ.
44, 101466 (2024)

\bibitem{lf11} A. Paliathanasis, Cosmological Constant from Equivalent
Transformation in Quantum Cosmology, (2024) [arXiv:2405.20683]

\bibitem{lf2} S. Dhasmana, A. Sen and Z.K. Silagadze, Equivalence of a
harmonic oscillator to a free particle and Eisenhart lift, Annals of Physics
434, 168623 (2021)

\bibitem{fam} F.M. Mahomed and P.G.L. Leach, Symmetry Lie Algebras of nth
Order Ordinary Differential Equations, J. Math. Anal. Appl. 151, 80 (1990)

\bibitem{noe} E. Noether, Invariante Variationsprobleme Koniglich
Gesellschaft der Wissenschaften Gottingen Nachrichten Mathematik-physik
Klasse 2, 235-267 (1918)

\bibitem{noe1} A. Halder, A. Paliathanasis and P.G.L. Leach, Noether's
Theorem and Symmetry, Symmetry 10, 744 (2018)

\bibitem{Leach80a} W. Sarlet, F.M. Mahomed and P.G.L Leach, Symmetries of
non-linear differential equations and linearisation, J. Phys A.: Math. Gen.
20, 277 (1987)

\bibitem{fam2} L.M. Berkovich, Factorization and Transformations of Ordinary
Differential Equations, Saratov Univ. Publ., Saratov, (1989)

\bibitem{fam3} R. Campoamor-Stursberg and J. Guer\'{o}n, Linearizing Systems
of Second-Order ODEs via Symmetry Generators Spanning a Simple Subalgebra,
Acta Appl. Math. 127, 105 (2013)

\bibitem{mel1} N.H. Ibragimov and S.V. Meleshko, Linearization of
third-order ordinary differential equations by point and contact
transformations, J. Math. Anal. Appl. 308, 266 (2005)

\bibitem{mel2} S.V. Meleshko, On linearization of third-order ordinary
differential equations, J. Phys. A.: Math.\ Gen. 39, 15135 (2006)

\bibitem{sun1} L.M. Berkovich, Method of Factorization of ordinary
differential operators and some of its applications, Appl. Anal. Disc. Math.
1, 122 (2007)

\bibitem{sun2} W. Nakpim and S.V. Meleshko, Linearization of third-order
ordinary differential equations by generalized Sundman transformations: $%
X^{\prime \prime \prime }+\alpha X=0$, Commun. Nonlinear Sci. Numer.
Simulat. 15, 1717 (2010)

\bibitem{sun3} A.Y. Al-Dweik, M.T. Mustafa, F.M. Mahomed and R.S. Alassar,
Math. Meth. Appl. Sci. 41, 6955 (2018)

\bibitem{fam4} G.W. Bluman, Potential Symmetries and Linearization. In:
Clarkson, P.A. (eds) Applications of Analytic and Geometric Methods to
Nonlinear Differential Equations. NATO ASI Series, vol 413. Springer,
Dordrecht (1993)

\bibitem{wafo} C. Wafo Soh, F.M. Mahomed, Linearization criteria for a
system of second-order ordinary differential equations, Int. J. Non-Linear
Mechanics 36, 671 (2001)

\bibitem{patera} J. Patera, R. T. Sharp, P. Winternitz, and H. Zassenhaus,
Invariants of real low dimension Lie algebras, J.\ Math. Phys. 17, 986 (1976)

\bibitem{ali1} S. Ali, F.M. Mahomed and A. Qadir, Linearizability criteria
for systems of two second-order differential equations by complex methods,
Nonlinear Dynamics 66, 77 (2011)

\bibitem{ali2} S. Ali, M. Safdar and A. Qadir, Linearization from complex
Lie point transformations, J. Appl. Math. 2014, 793247 (2014)

\bibitem{ali3} M.\ Safdar,\ S. Ali and F.M. Mahomed, Linearization of
systems of four second-order ordinary differential equations, Pranama - J.
Phys. 77, 581 (2011)

\bibitem{qa1} A. Qadir, Geometric Linearization of Ordinary Differential
Equations, Sigma 3, 1 (2007)

\bibitem{ami} A.V. Aminova, Projective transformations and symmetries of
differential equations, Sbornik Math. 185, 1711 (1995)

\bibitem{mt1} M. Tsamparlis, Linearization of Second-Order Non-Linear
Ordinary Differential Equations: A Geometric Approach, Symmetry 15, 2082
(2023)

\bibitem{Rya1} M.P. Ryan and L.C. Shepley, Homogeneous Relativistic
Cosmologies, Princeton University Press, Princeton, (1975)

\bibitem{cont1} P.A. Terzis, N. Dimakis, T. Christodoulakis, A.
Paliathanasis and M.\ Tsamparlis, Variational contact symmetries of
constrained Lagrangians, J. Geom. Phys. 101, 52 (2016)

\bibitem{axioms} A. Paliathanasis, Solving Nonlinear Second-Order ODEs via
the Eisenhart Lift and Linearization, Axioms 13, 331 (2024)

\bibitem{comm} A.\ Paliathanasis, The Common Solution Space of General
Relativity, (2024) [arXiv:2406.01998]

\bibitem{ndl} T. Christodoulakis, N. Dimakis and P.A.\ Terzis, Lie point and
variational symmetries in minisuperspace Einstein gravity, J. Phys. A: Math.
Theor. 47, 095202 (2014)

\bibitem{ansym1} M. Tsamparlis and A. Paliathanasis, Two-dimensional
dynamical systems which admit Lie and Noether symmetries, J. Phys. A: Math.
Theor. 44, 175202 (2011)

\bibitem{el2} A. Galajinsky and I.\ Masterov, Eisenhart lift for higher
derivative systems, Phys. Lett. B 765, 86 (2017)

\bibitem{mtgrg} M. Tsamparlis, A. Paliathanasis and L. Karpathopoulos, Exact
solutions of Bianchi I spacetimes which admit Conformal Killing vectors,
Gen. Rel. Grav. 47, 15 (2015)

\bibitem{anszek} A. Paliathanasis and P.G.L. Leach, Symmetries and
singularities of the Szekeres system, Phys. Lett. A 381, 1277 (2017)

\bibitem{qq} A. Paliathanasis, A. Zampeli, T. Christodoulakis and M.T.
Mustafa, Quantization of the Szekeres system, Class. Quantum Grav. 35,
125005 (2018)

\bibitem{szek0} P. Szekeres, A class of inhomogeneous cosmological models,
Commun. Math. Phys. 41, 55 (1975)

\bibitem{andrez} A. Krasinski, Inhomogeneous cosmological models, Cambridge
Univ. Press, Cambridge, UK, (2011)

\bibitem{ccbar} J.D. Barrow and J. Stein-Schabes, Inhomogeneous cosmologies
with cosmological constant, Phys. Lett. A 103, 315 (1984)

\bibitem{zzbar} A. Zampeli and A. Paliathanasis, Quantization of
inhomogeneous spacetimes with cosmological constant term, Class. Quantum
Grav. 38, 165012 (2021)

\bibitem{dim1} T. Christodoulakis, N. Dimakis, P.A. Terzis, B. Vakili, E.
Melas and Th. Grammenos, Minisuperspace canonical quantization of the
Reissner-Nordstr\"{o}m black hole via conditional symmetries, Phys. Rev. D
89, 044031 (2014)

\bibitem{sch3} H. Reissner, \"{U}ber die Eigengravitation des elektrischen
Feldes nach der Einsteinschen Theorie, Annals of Physics 50, 106 (1916)

\bibitem{sch4} G. Nordstr\"{o}m, On the Energy of the Gravitational Field in
Einstein's Theory, Proceedings Series B Physical Sciences 26, 1201 (1918)

\bibitem{anpo} A. Paliathanasis, Projective Collineations of Decomposable
Spacetimes Generated by the Lie Point Symmetries of Geodesic Equations,
Symmetry 13, 1018 (2021)

\bibitem{yam1} J.M. Lee and T.H. Parker, The Yamabe problem, Bull (NS) Am
Math Soc 17, 37 (1987)

\bibitem{yam2} B. Abraham-Shrauner, K.S. Govinder and D.J. Arrigo, Type-II
hidden symmetries of the linear 2D and 3D wave equations, J. Phys. A. Math.
Gen. 39, 5739 (2006)
\end{thebibliography}
\end{document}